\documentclass[12pt,a4paper]{article}

\usepackage{authblk}

\usepackage{graphicx}

\usepackage{amstext}
\usepackage{amsmath}
\usepackage{amsfonts}
\usepackage{subcaption}

\usepackage{mathtools}
\usepackage{float}
\usepackage{fancyhdr}

\usepackage[utf8]{inputenc}
\usepackage[english]{babel}
\usepackage{amsthm}

\newtheorem{theorem}{Theorem}[section]
\newtheorem*{definition}{Definition}

\newtheorem{corollary}[theorem]{Corollary}
\newtheorem{lemma}[theorem]{Lemma}

\usepackage[margin=1.0in]{geometry}
\title{A manifestly covariant framework for causal set dynamics}

\author[1,2]{Fay Dowker}
\author[1]{Nazireen Imambaccus}
\author[1]{Amelia Owens}
\author[2]{Rafael Sorkin}
\author[1]{Stav Zalel}
\affil[1]{Blackett Laboratory, Imperial College London, SW7 2AZ, U.K.}
\affil[2]{Perimeter Institute, 31 Caroline Street North, Waterloo ON, N2L 2Y5, Canada.}

\usepackage{hyperref}

\newcommand{\lc}[1]{\tilde{#1}}

\def\twoch{\,\,\begin{picture}(0,1) 
\thicklines
\multiput(0,0)(0,10){2}{\circle*{2}}
\put(0,0){\line(0,1){10}}
\end{picture}\,\,}

\def\twoach{\,\begin{picture}(1,1) 
\thicklines
\multiput(0,0)(10,0){2}{\circle*{2}}
\end{picture}\,\,\,}

\def\Lcauset{\,\,\begin{picture}(2,2) 
\thicklines
\multiput(0,0)(0,10){2}{\circle*{2}}
\put(0,0){\line(0,1){10}}
\multiput(0,0)(10,0){2}{\circle*{2}}
\end{picture}\,\,}

\begin{document}

\maketitle

\abstract{We propose a manifestly covariant framework for causal set dynamics. The framework is based on a structure, dubbed \textit{covtree}, which is a partial order on certain sets of finite, unlabeled causal sets. We show that every infinite path in covtree corresponds to at least one infinite, unlabeled causal set. We show that transition probabilities for a classical 
random walk on covtree induce a classical 
measure on the $\sigma$-algebra generated by the stem sets.}

\tableofcontents

\section{Introduction}

General Relativity (GR) has the property of general covariance.  As discussed in section 7 of \cite{Sorkin:2007qh}, this gauge invariance of GR has two facets. The first is that physical statements -- or, equivalently, properties or predicates or events -- in GR must be diffeomorphism invariant. The second is that the equations of motion of GR are diffeomorphism invariant so that metrics in a diffeomorphism equivalence class must either all satisfy or all not satisfy the equations of motion. 
As stated in \cite{Sorkin:2007qh}, ``[In GR] the second facet of covariance flows directly from the first as a consistency condition, because it would be senseless to identify two metrics one of which was allowed by the equations of motion and the other of which was forbidden; and conversely, the kinematical identification must be made if one wishes the dynamics to be deterministic. Thus, the first or ``ontological'' facet of general covariance tends to coalesce with its second or ``dynamical'' facet.''

There is a widespread expectation that GR will turn out to be an approximation, at large scales and in certain circumstances, to a deeper theory of quantum gravity.  A. Einstein's struggles over the understanding of general covariance were central to the development of GR and one might expect that grappling with the corresponding issues within quantum gravity will be important to its development too \cite{Stachel:2014}. The requirement that the diffeomorphism invariance of GR must emerge from quantum gravity in the large scale approximation is indeed used to formulate guiding precepts for each approach  to quantum gravity, though the form that these precepts take varies from approach to approach. 

In the case of the causal set approach to the problem of quantum gravity,  any deeper precept of general covariance cannot  be literally diffeomorphism invariance because diffeomorphism is a continuum concept and causal set theory is discrete. Causal set theory postulates that the fundamental structure of spacetime is atomic at the Planck scale and takes the form of a causal set or locally finite partial order.\footnote{At least, this is the kinematics of the theory. In the full quantum theory this statement will be revised to take account of quantum interference between causal sets in the sum over histories.}  The elements of the causal set are the atoms of spacetime and continuum spacetime  is an approximation to the causal set at large scales. The order relation of the underlying causal set reveals itself as the causal structure of the approximating continuum spacetime and the number of causal set elements manifests itself as the spacetime volume of the approximating continuum spacetime. For a recent review of causal set theory see \cite{Surya:2019ndm}. 

The structure of continuum spacetime, then, emerges from Order and Number and this central conjecture of causal set theory has an immediate consequence: the physical content of a causal set is independent of what mathematical objects the causal set elements are and is also independent of any additional labels those causal set elements might carry: only the order relation of the elements
 and the number of elements has physical meaning. This ``mathematical-identity-and-label independence'' is a good candidate for a condition of general covariance in causal set theory, at least as far as the first facet, mentioned above, goes \footnote{Whether this condition alone is enough to give rise to diffeomorphism invariance is subsumed in the question of whether causal sets can give rise to a continuum approximation at all.}. 

This form of general covariance is a consequence of a more general guiding heuristic, Occam's Razor, applied in the particular case of causal set quantum gravity and grounded in earlier seminal work, theorems in Lorentzian geometry by Penrose, Kronheimer, Hawking and Malament \cite{Kronheimer:1967, Hawking:2014xx, Malament:1977}. These theorems show that the spacetime causal order plus spacetime volume are sufficient, in the continuum, to provide the full geometry of a Lorentzian spacetime for a very large class including all globally hyperbolic spacetimes. This is strong evidence that order and number in the discrete substructure are together sufficient to encode approximate Lorentzian geometrical information at large scales. From this then arises the principle that the mathematical identity of the spacetime atoms is not physical. 

The second facet of general covariance in GR mentioned above, in which all diffeomorphic manifold-metric pairs either do or do not satisfy the Einstein equations, does not have such a direct analogue in causal set theory as it currently stands. The dynamical models developed thus far are stochastic and there is no analogue of ``equations of motion'' that a given causal set can either satisfy or not in a binary distinction. Nevertheless, a specific proposal for a condition of ``discrete general covariance'' was made in the context of a particular paradigm for causal set dynamics, namely classical random models of causal set growth, and this condition was used to construct an interesting family of classical stochastic dynamical models for causal sets, the Classical Sequential Growth (CSG) models \cite{Rideout:1999ub}. Each CSG model is a stochastic process of growth of the causal set spacetime in which new elements are born in sequence, forming relations with the previous elements in the sequence at random with a probability distribution given by the particular model. The sequence of the births is a total order on the spacetime atoms and contains unphysical, gauge information. As described further in section \ref{dynamics_for_causets_subsec}, the definition of a CSG model is given in terms of this sequence, as is the discrete general covariance condition\footnote{The discrete general covariance (DGC) condition could be considered as the ``dynamical'' facet within causal set theory. The DGC condition imposes that, given a pair of order-isomorphic causal sets with cardinality $n$, the probabilities of growing each by stage $n$ are equal. This is somewhat akin to the ``dynamical'' facet in GR where all diffeomorphic manifold-metric pairs have the same action and therefore the same weight in the path integral.}.  This condition is well-justified within the context of the framework of sequential growth but the framework  depends on and refers to the unphysical sequential label of ``stage''. 
The question arises: is it possible to define physically interesting causal set growth dynamics that only ever refer to the physical degrees of freedom, to the physical partial order with no reference to any other labels? This paper describes work motivated by this question and provides a positive
first step in that direction.

\section{Review} \label{CSGM}

\subsection{Preliminaries}

In this subsection we list some of the terminology and assumptions used in this paper.  A more complete glossary of causal set terminology is given in \cite{glossary}. All the infinite causal sets we consider in this paper are countable and past-finite (see below). 

\noindent  Let $(C, \prec)$ be a causal set. We use the irreflexive convention in which $x \not\prec x$. The word
causet is short for causal set.

\begin{itemize}

\item[] If $x\prec y$ we say $x$ is below $y$, $y$ is above $x$, $x$ is an ancestor of $y$ or $y$ is a descendant of $x$.

\item[] The \textbf{past} of $x\in C$ is the subcauset $past(x):=\{y\in C | y\prec x\}$. This is the non-inclusive past:
$x \not\in past(x)$. The \textbf{{future}} of $x$ is defined similarly.

\item[] $C$ is \textbf{past-finite} if $|past(x)|<\infty\,, \ \forall \ x \in C$. 

\item[] A \textbf{stem} in $C$ is a finite subcauset $S$ of $C$ such that  if $x\in S$ and $y\prec x$ then $y \in S$.
 \item[] An $\mathbf{n}$-\textbf{stem} is a stem with cardinality $n$. 

\item[] A relation $x \prec y$ is called a \textbf{link} if there is no element in the order
between $x$ and $y$. In that case we say $x$ is directly below $y$, $y$ is directly above $x$, $x$ is a direct ancestor of $y$ or $y$ is a direct descendant of $x$. A link is also called a covering relation and we can also say $y$ covers $x$. 

\item[] 
An \textbf{antichain} is a causet whose elements are unrelated to each other. 

\item[]  A \textbf{chain} is a causet whose elements are all related. 

\item[] A \textbf{path} in $C$ is a subcauset of $C$ which is a chain all of whose links are also links of $C$.

\item[] The element $x$ of $C$ is in \textbf{level} $\mathbf{n}$ if the
longest chain of which $x$ is the maximal element 
has cardinality $n$. 
Level 1 of $C$ comprises the minimal elements of $C$, 
level 2 comprises the minimal elements of what remains of $C$ after the elements in level 1 are deleted, \textit{etc}.

\item[] It is useful to represent a causet as a graph in a \textbf{Hasse diagram} in which elements are represented by nodes, and there is an upward-going edge from $x$ to $y$ if and only if $x \prec y$ is a link. The other relations are implied by transitivity. All the pictures of causal sets in this paper are Hasse diagrams. 

\item[] An \textbf{isomorphism}, $f$, between two causets, $C, D$, is a bijection $f:C\rightarrow D$ such that $f(x)\prec_D f(y) \iff x\prec_C y, \ \forall \ x,y \in C$.  If $C$ and $D$ are isomorphic,  we write $C\cong D$.

\end{itemize}
\subsection{Labeled causets and \textit{n}-orders}

For concreteness, we fix a collection of \textit{ground sets} for the causal sets we will work with in this paper, following notation and terminology adapted from \cite{Brightwell:2002vw, Ash:2005za}. Consider, for each $n>0$,  the interval of natural numbers, $[n-1]:=\{0,1,...,n-1\}$ of cardinality $n$.  Let $\tilde{\Omega}(n)$ denote the set of partial orders, $\prec$, on the ground set $[n-1]$ satisfying $i\prec j \implies i<j$. 
We call an element of $\tilde{\Omega}(n)$ a \textbf{finite labeled causet}.\footnote{More generally, for a causal set $C$ of cardinality $n$ we call a bijection  $f:[n-1] \rightarrow C$ 
a natural labeling if $f(i) \prec f(j)\implies i < j$.  Using the interval $[n-1]$ itself as the ground set for $C$, 
together with the condition that $i\prec j \implies i<j$, 
makes the identity map into a natural labeling.} 

 We define $\tilde{\Omega}(\mathbb{N}):=\bigcup \limits_{n\in\mathbb{N}^+}\tilde{\Omega}(n)$, the set of all finite labeled causets. 

We define $\tilde{\Omega}$ to be the set of partial orders on the ground set $\mathbb{N}$ satisfying $i\prec j \implies i<j$. 
We call an element of $\tilde{\Omega}$ an \textbf{infinite labeled causet}.  By this definition, every infinite labeled causet is past finite since element $j$ can be above at most $j$ other elements. The converse is also true: any infinite, past finite causet admits a natural labeling by the natural numbers \cite{Brightwell:2011}.

We denote labeled causets (finite or infinite) and their stems with a tilde.
Figure \ref{stem_in_lc} gives some examples of stems of a labeled causets. By our definitions, not all stems of a labeled causet are 
themselves labeled causets because the ground set of the stem may not itself be an interval, as shown in figure \ref{stem_in_lc}. 
\begin{figure}[htpb]
  \centering
	\includegraphics[width=0.74\textwidth]{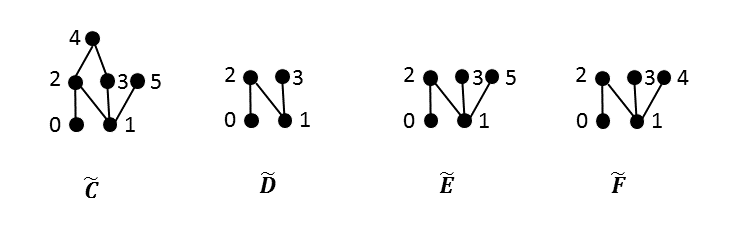}
	\caption{$\lc{C}, \lc{D}$ and $\lc{F}$ are labeled causets. $\lc{D}$ and $\lc{E}$ are stems in $\lc{C}$. $\lc{F}$ is \textit{not} a stem in $\lc{C}$ because it is not a subcauset of $\lc{C}$. $\lc{E}$ is {not} a labeled causet by our definition because its ground set is not an interval of integers.}
	\label{stem_in_lc}
\end{figure}

As described in the introduction, it is a tenet of causal set theory that the atoms of spacetime have no structure.  It is of no physical relevance what mathematical objects the elements of a causal set are.  One way to express this is to say we are ultimately interested only in isomorphism equivalence classes of causal sets. 

Isomorphism is an equivalence relation on each $\tilde{\Omega}(n)$, and on $\tilde{\Omega}$. We define \textbf{unlabeled causets}, or \textbf{orders} for short, to be isomorphism classes of labeled causets. 
An \textbf{unlabeled causet of cardinality} $\mathbf{n}$, or $\mathbf{n}$-\textbf{order} for short,
 is an isomorphism class, $C=[\lc{C}] =\{\lc{D}\in \tilde{\Omega}(n)\, |\, \lc{D}\cong \lc{C}\}$, where $\lc{C}\in \tilde{\Omega}(n)$ is some representative of $C$.

An \textbf{infinite unlabeled causet}, or \textbf{infinite order} for short, is an isomorphism class, $C=[\lc{C}] =\{\lc{D}\in \tilde{\Omega}\, |\, \lc{D}\cong \lc{C}\}$, where $\lc{C}\in \tilde{\Omega}$ is some representative of $C$. 
We define $\Omega(n)$ to be the set of $n$-orders, and ${\Omega}(\mathbb{N}):=\bigcup \limits_{n\in\mathbb{N}^+}{\Omega}(n)$
is the set of finite orders. $\Omega$ is defined to be the set of infinite orders.

We generalise the concept of stem to orders. 
We say a finite order, $S$, is a stem in order $C$ if there exists a representative of $S$ 
which is a stem in a representative of $C$ and in this case we say, variously, $S$ is a stem in $C$, or $S$ occurs as a stem in $C$ or $C$ contains $S$ as a stem. 
We say 
a finite order, $S$, is a stem in labeled causet $\lc{C}$ if the order $S$ is a stem in the order $[\lc{C}]$. 
 So, the meaning of stem depends on the context. 

 Examples are shown in figure \ref{defn_unlabeled_stem}. 
Note that if order $S$ is a stem in order $T$ and $T$ is a stem in order $U$ then $S$ is a stem in $U$: ``a stem in a stem is a stem'' . So there is an ``order-by-inclusion-as-stem'' on the set of all finite orders.

\begin{figure}[htpb]
  \centering
	\includegraphics[width=0.64\textwidth]{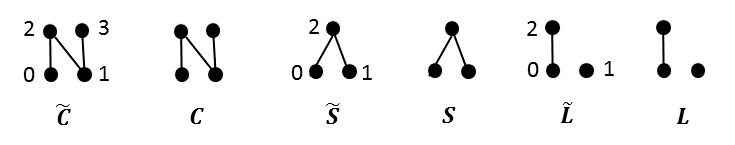}
	\caption{Labeled causets  $\lc{C}, \ \lc{S}$ and $\lc{L}$ are representatives of orders $C,\ S$ and $L$, respectively. $\lc{S}$ is a 3-stem in $\lc{C}$. $\lc{L}$ is not a subcauset of $\lc{C}$ so it is not a stem in $\lc{C}$. $S$ and $L$ are 3-stems in $\lc{C}$ and in $C$. }
	\label{defn_unlabeled_stem} 
\end{figure}

Finally, we introduce a concept that will be important later. 
An infinite order $C\in \Omega$ is a \textbf{rogue} \cite{Brightwell:2002vw} if there exists an infinite order $D$ such that $D \neq C$ and the two orders have the same stems. If infinite orders $C$ and $D$ have the same stems we write $C\sim_R D$. If $C\sim_R D$ and $C\not= D$, we say that $C$ and $D$ are equivalent rogues. An example of a pair of equivalent rogues is given in figure \ref{rogueexample}.

\begin{figure}[htpb]
  \centering
	\includegraphics[width=0.44\textwidth]{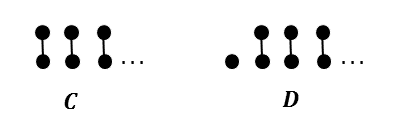}
	\caption{$C$ is a countable union of 2-chains and $D$ is the union of $C$ with a single unrelated element.  $C$ and $D$ have the same stems -- any union of finitely many 2-chains and a finite, unrelated antichain  -- 
	so $C$ and $D$ are equivalent rogues.}
	\label{rogueexample}
\end{figure}

\subsection{Dynamics}
\label{dynamics_for_causets_subsec}

\par Guided by the insight that path integral quantum theory is a form of generalised measure theory \cite{Sorkin:1994dt}, and by the heuristic of \textit{becoming} \cite{Sorkin:2007qh, spacetimeatoms}, a major breakthough in the development of a dynamics for causal sets was the construction of the Classical Sequential Growth (CSG) models by Rideout and Sorkin
\cite{Rideout:1999ub}. 
A CSG model is a  stochastic process consisting of the sequential coming into being, or \textit{birth}, of new causet elements and the formation of relations between each newly born element and a randomly chosen subset of the elements  born previously in the sequence. 
The process can be represented as an upward-going random walk on a partially ordered tree called labeled poscau (short for the poset of labeled causal sets):

\begin{definition}\par \textbf{Labeled poscau} is the partial order $(\tilde{\Omega}(\mathbb{N}), \prec)$,  where $\lc{S}\prec \lc{R}$ if and only if $\lc{S}$ is a stem in $\lc{R}$.\footnote{We use the symbol $\prec$ to denote the relation for several different partial orders in this work. The meaning of $\prec$ in each case is to be inferred from the context.}
\end{definition}

Labeled poscau is a tree formed of countably many levels, the first three of which are shown in figure \ref{lposcau}. 
A growth model based on labeled poscau is a random walk formed of a sequence of stages. At stage $n$, a labeled causet, $\tilde{C}_n$, of cardinality $n$ transitions to a labeled causet, $\tilde{C}_{n+1}$, of cardinality $n+1$ such that $\tilde{C}_n$ is a stem in $\tilde{C}_{n+1}$. The transition can be thought of as the birth of the new element $n$, of $\tilde{C}_{n+1}$, which comes into being above a randomly chosen subset of $\tilde{C}_n$, the \textit{ancestor set} of the newborn.  The labeled causet $\tilde{C}_{n+1}$
 is one of the `child' causets that are directly above the `parent' $\tilde{C}_n$ 
 in labeled poscau. Any assignment of transition probabilities, satisfying the Markov sum rule, to all the links in labeled poscau gives a  well-defined stochastic process. Each infinite path, $\lc{C}_1\prec \lc{C}_2\prec ... \prec \lc{C}_k\prec\dots$, in labeled poscau beginning at the root
 is identified with the infinite labeled causet $\bigcup\limits_{i\in\mathbb{N}^+}\lc{C}_i$.
This is a one-to-one correspondence and so the histories in the model can be though of, equivalently, as elements of $\tilde{\Omega}$ or as infinite paths in labeled poscau.  
   
 An \textit{event} in such a stochastic process is a measureable subset of $\tilde{\Omega}$.  
For example, corresponding to each node, $\lc{C}_n$ of cardinality $n$,  in labeled poscau is
 the cylinder set \cite{Brightwell:2002vw},  
\begin{align}
cyl(\lc{C}_n) := \{  \lc{D} \in \tilde{\Omega} \ |\  \lc{D}|_{[n-1]}   = \lc{C}_n \}\,. 
\end{align}
The measure of each cylinder set is given by 
$\tilde{\mu}(cyl(\lc{C_n}))=\mathbb{P}(\lc{C_n})$, where 
$\mathbb{P}(\lc{C_n})$ is the probability that the random walk reaches $\lc{C_n}$. This measure can be uniquely 
extended to a measure on the $\sigma$-algebra $\tilde{\mathcal{R}}$ 
generated by the collection of cylinder sets,  by standard results in stochastic processes and measure theory \cite{Kolmogorov:1975}.  

Now, not all events in $\tilde{\mathcal{R}}$ are physical because they are not all covariant. For example, the cylinder set $cyl(\lc{C}_n)$ is the event ``the causet at the end of stage $n-1$ of the process is $\lc{C}_n$'' which refers to the unphysical, gauge information of the stage. An event, $\mathcal{E}$, is covariant if whenever a labeled causet, $\lc{C}$, is in $\mathcal{E}$ then all labeled causets isomorphic to $\lc{C}$ are also in $\mathcal{E}$. $\mathcal{E}$ can then be identified, in an obvious way, with a set of orders. We define the sub-$\sigma$-algebra, $\mathcal{R}\subset \tilde{\mathcal{R}}$,     as the algebra of all covariant measureable events \cite{Brightwell:2002yu, Brightwell:2002vw}.

\begin{figure}[htpb]
  \centering
	\includegraphics[width=0.6\textwidth]{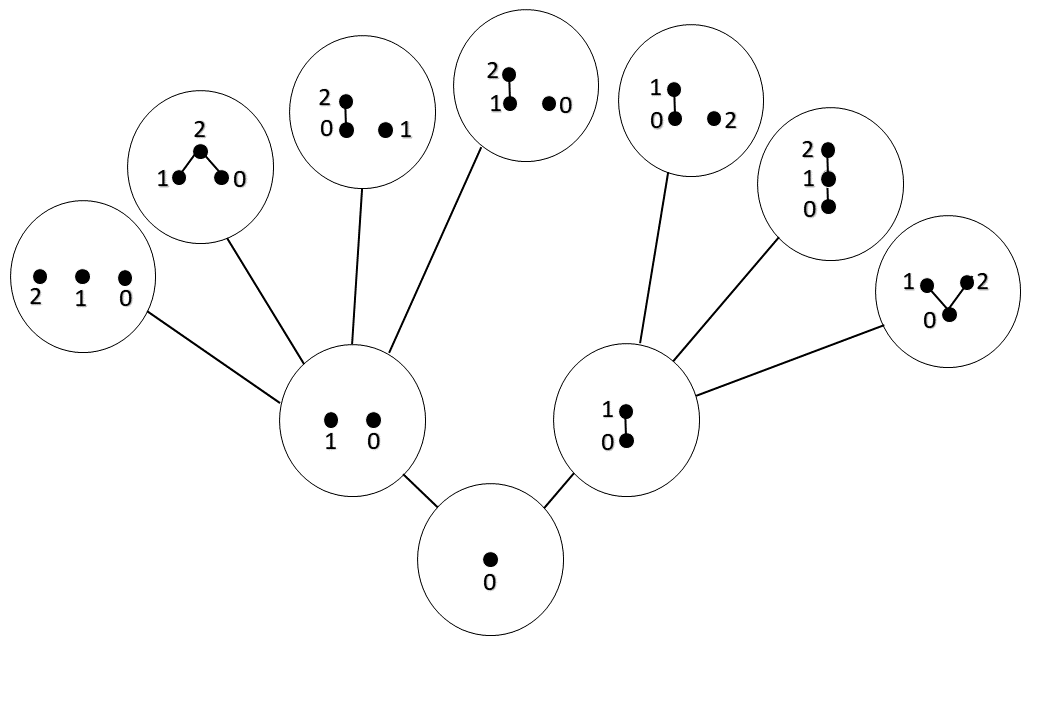}
	\caption{The first three levels of labeled poscau. The random walk starts at the root and proceeds upwards.}
	\label{lposcau}
\end{figure}
 
\subsection{Classical Sequential Growth models}

The collection of random walks up labeled poscau is vast, so Rideout and Sorkin (RS) imposed physically motivated conditions to restrict the models to a more interesting class, the
 Classical Sequential Growth (CSG) models.  The transition probabilities for CSG models were derived by RS by imposing on the random walk  two conditions:  Bell Causality (BC) and Discrete General Covariance (DGC) \cite{Rideout:1999ub}. 
DGC is the condition that the probability of arriving at any node of labeled poscau depends only on the isomorphism class of the node. For example, the probabilities to arrive at the three nodes in figure \ref{lposcau} which are 
in the isomorphism class of the $``L"$ 3-order (\Lcauset \ \ ) are equal in a CSG model. BC is analogous to the local causality condition that enters in the derivation of the Bell inequalities in Bell's no-local-hidden-variables theorem. At stage $n$, consider 
two possible transitions from a parent causet $\lc{C}$ either to child $\lc{A}$ or to child $\lc{B}$. Suppose there is an element, 
$k$, of $\lc{C}$, which is not in the ancestor set of the newborn element $n$, neither in $\lc{A}$ nor in $\lc{B}$.
Such an element $k$ is called a \textit{spectator} of both transitions. Now consider transitions at stage $n-1$, 
$\lc{C'} \rightarrow \lc{A'}$ and $\lc{C'} \rightarrow \lc{B'}$, which are formed from the previous ones by deleting 
the spectator $k$ from $\lc{C}$, $\lc{A}$ and $\lc{B}$ and consistently relabeling the remaining causal sets so their base sets are integer intervals. Bell Causality is the condition
\begin{align}
\frac{  \mathbb{P}(\lc{C} \rightarrow \lc{A})  } {  \mathbb{P}(\lc{C} \rightarrow \lc{B})  }
= \frac{  \mathbb{P}(\lc{C'} \rightarrow \lc{A'})  } {  \mathbb{P}(\lc{C'} \rightarrow \lc{B'})  }\,.
\end{align}

RS showed that these two conditions of BC and DGC imply that a CSG model is specified by a sequence of non-negative real numbers, $\{t_0, t_1, t_2, \dots\}$, which determine the transition probability for each possible transition 
$\tilde{C}_n \rightarrow \tilde{C}_{n+1}$ in the following way. The newly born element $n$ chooses a subset $Y$ from amongst all the subsets of $\tilde{C}_n$ with relative probability $t_{|Y|}$ and $n$ is put above all elements of $Y$ and the transitive closure taken.  For completeness, we give the explicit form of the transition amplitude in a CSG model for the transition $\tilde{C}_n \rightarrow \tilde{C}_{n+1}$:
\begin{align}\label{transprob}
\mathbb{P}(\tilde{C}_n \rightarrow \tilde{C}_{n+1}) &= \frac{\lambda(\varpi, m)}{\lambda(n, 0)}\,,
\end{align}
where 
$\varpi$ is the cardinality of the ancestor set of the newborn element $n$, $m$ is the number of maximal elements of 
the ancestor set of  $n$ and  
\begin{equation} 
\lambda(k, p) := \sum_{i = 0}^{k - p} \binom{k-p}{i} t_{p+i}.
\end{equation}

The covariant events in CSG models were fully characterised and given a physical interpretation in \cite{Brightwell:2002yu,Brightwell:2002vw}. Here we give a brief summary of those results which were based on the 
concept of stem set.  Given an $n$-order $C_n$, the stem set 
$stem(C_n)$ is the event ``$C_n$ is a stem in the growing order'' and is given mathematically by 
\begin{align}
stem(C_n) := \{  \lc{D} \in \tilde{\Omega} \ |\  C_n \ \textrm{is a stem in} \ \lc{D} \}\,. 
\end{align}

Let $\mathcal{S}$ denote the collection of all stems sets, and let $\mathcal{R}(\mathcal{S})$ denote the \textit{stem algebra}, the $\sigma$-algebra generated by $\mathcal{S}$. We call an element of $\mathcal{R}(\mathcal{S})$ a \textit{stem event}. Stem events are covariant, and $\mathcal{R}(\mathcal{S})$ is a sub-$\sigma$-algebra of $\mathcal{R}$.
This inclusion is strict, mathematically, but in a well-defined, physical sense the stem algebra exhausts all the covariant events. Indeed, for every covariant event, $\mathcal{E}$,  one can find a stem event $\mathcal{F}$ such that the symmetric difference between $\mathcal{E}$ and $\mathcal{F}$ is a set of rogue causets which is of measure zero in any CSG model because the set of all rogues is of measure zero.
 In other words, in CSG models covariant events are, for all practical purposes, 
 stem events. This is important. It means that every physical statement in a CSG model for which the dynamics provides a probability is a (countable) logical combination  of statements about which finite orders are stems in the causet
universe.

\section{A covariant framework} \label{seccovtree}

There exist several successful gauge theories, including GR, that are defined mathematically in terms of their gauge dependent degrees of freedom but in which physical statements can be made, purged of any unphysical gauge dependence introduced along the way.
CSG models make sense physically in this way. Although the definition of a CSG model is given in terms of an unphysical sequence of birth events, the model provides an exhaustive set of physically comprehensible, covariant measureable events from which we make physical predictions. Were we only seeking a class of interesting classical growth models to explore, we might be content with CSG models as we have them. But for quantum gravity in the causal set approach, the task in hand is to find a quantum dynamics for causal sets, from which GR must then emerge as a large scale approximation. We seek quantum growth models. 

One possible route to a Quantum Causet Dynamics would be to try to generalise what was done for CSG to the quantum case, finding appropriate analogues of the DGC and BC conditions on a decoherence functional or double path integral for a growing causal set \cite{Dowker:2010qh}.  In this paper,  we take a slightly different path by asking whether there is an explicitly label-independent framework for
classical causet growth,  an alternative to labeled poscau, which might suggest novel possibilities for quantal generalisations. We frame the question as:  is it possible to construct a physically well-motivated measure on the stem algebra $\mathcal{R}(\mathcal{S})$ 
\textit{directly}, in a manifestly label-independent way that does not rely on any gauge dependent notion and which respects the heuristic of growth and becoming? 

There already exists a structure in the literature, poscau  \cite{Rideout:1999ub}, which at first sight might seem to furnish such a framework.  Poscau is a partial order on finite orders, 
$(\Omega(\mathbb{N}), \prec)$, where $A\prec B$ if and only if $A$ is a stem in $B$.  Figure \ref{hasse} shows a Hasse diagram of the first three levels of poscau.\footnote{Rideout and Sorkin originally used poscau to introduce CSG models.}
\begin{figure}[htpb]
 \centering
	\includegraphics[width=0.6\textwidth]{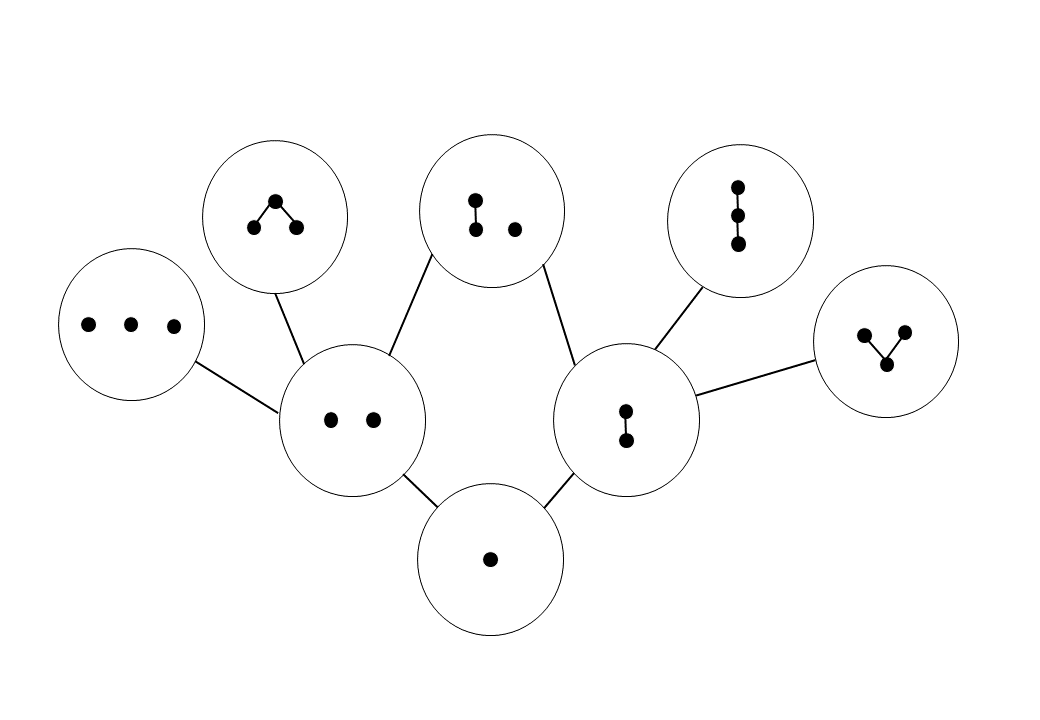}
	\caption{The first three levels of poscau.}
	\label{hasse}
\end{figure}
If one identifies each node $A$ in poscau with its stem set, $stem(A)$, one might be tempted to try to define a dynamics as a random walk up poscau such that arriving at the node $A$ corresponds to the occurrence of the covariant event $stem(A)$. This does not work because in such a dynamics only one stem set at each level can occur and the growing order would only have a single stem of each finite cardinality. 

Thinking in this way, however, suggests the solution: the walk should be on a tree formed of countably many levels in which the nodes in level $n$ are not single $n$-orders but \textit{sets} of $n$-orders. Each set of 
$n$-orders in level $n$ will correspond to the covariant event ``the $n$-stems of the growing order are the elements of this set.''  We  will call this tree  \textbf{covtree} (short for \textbf{covariant tree})
and in the classical case, the dynamics will take the form of a stochastic process consisting 
of a sequence of 
stages, each of which is a transition from a node in one level of covtree to 
a node in the level above, just as the CSG models are defined on labeled 
poscau.  We will make this precise in the rest of this section. 

\subsection{Certificates}

Let $\Gamma_n$ be a (non-empty) set of $n$-orders, \textit{i.e.} a subset of $\Omega(n)$. 
\begin{definition}
An order $C$ is a \textbf{certificate} of $\Gamma_n$ if $\Gamma_n$ is the set of  all $n$-stems in $C$.
\end{definition}

\par 
Note that a given $\Gamma_n\subseteq \Omega(n)$ may have no certificate: see the example $\Gamma'_3$ in figure \ref{cert_defns}. We use $\Lambda$ to denote the collection of sets of $n$-orders, for all $n$, which have certificates:

\begin{align}
\Lambda := \bigcup\limits_{n\in\mathbb{N}^+}\{\Gamma_n \subseteq \Omega(n)| \exists \text{ a certificate for } \Gamma_n\}\,. 
\end{align}
Note also that if  $\Gamma_n$ has a certificate then it has infinitely many certificates and that if $\Gamma_n$ has a certificate then it has a finite certificate. 

\begin{figure}[htpb]
  \centering
	\includegraphics[width=0.66\textwidth]{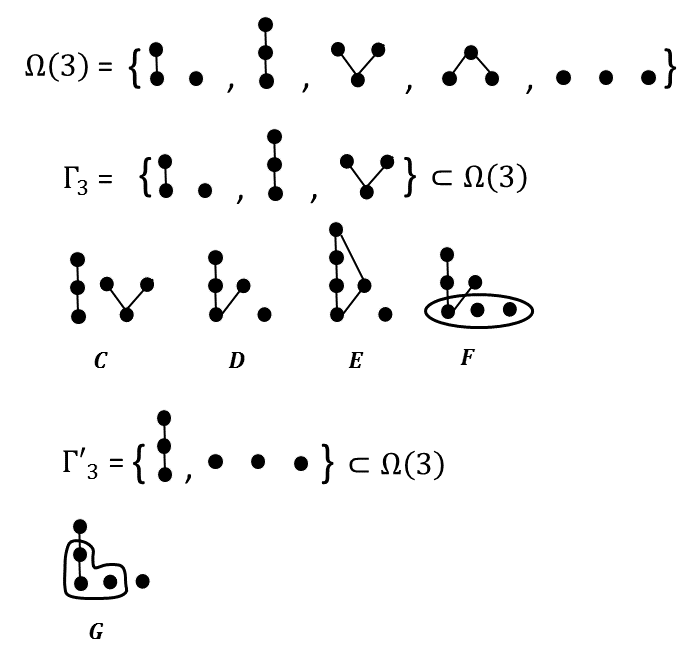}
	\caption{
	$\Omega(3)$ and two of its subsets, $\Gamma_3$ and $\Gamma'_3$, are shown. $C$, $D$ and $E$ are certificates of $\Gamma_3$. $F$ is \textit{not} a certificate of $\Gamma_3$ because $F$ contains the 3-antichain (circled in the figure) as a 3-stem. $\Gamma'_3$ has no certificates because every order which contains the 3-chain and the 3-antichain also contains the ``L'' 3-order as illustrated by $G$.
	}
	\label{cert_defns}
\end{figure}

\begin{definition} Given some $\Gamma_n\in\Lambda$, we order its finite certificates as follows: let $C_1,C_2$ be finite certificates of $\Gamma_n$, then $C_1\preceq C_2$ if $C_1$ is a stem in $C_2$. A \textbf{minimal certificate} of $\Gamma_n$ is minimal in this order.
\end{definition}

\par If $\Gamma_n\in\Lambda$ has more than one minimal certificate, these minimal certificates need not have the same cardinality\footnote{Recall that we define the cardinality, $|C|$, of an $n$-order $C$ as $|C|:=n$.} as each other. Also, an $n$-order in $\Gamma_n$ may be embedded in a minimal certificate of $\Gamma_n$ in more than one way.\footnote{This is shorthand for a more precise statement. Let the certificate of $\Gamma_n$ be $C$ and let the $n$-order be $X\in \Gamma_n$. 
We say that $X$ can be embedded in $C$ in $k$ ways if, for any labeled representative $\lc{C}$  of $C$, there are $k$ different subcausets of $\lc{C}$ which are stems and which are isomorphic to a representative of $X$. } Examples are shown in figure \ref{min_cert_defns}. 

\begin{lemma} Let $\Gamma_n=\{A^1,A^2,...,A^k\}$ be a set of $n$-orders. If $C$ is a minimal certificate of $\Gamma_n$ then $n \le |C| \le kn$.
$|C|= n$ if and only if $\Gamma_n$ is a singleton set ($k = 1$). 
\label{cardmincert}
\end{lemma}
\begin{proof}
Consider a labeled representative $\lc{C}$ of $C$. 
For each $A^i$, $i = 1,2,\dots k$,  take a subset of $\lc{C}$ that is a stem in $\lc{C}$, 
isomorphic to $A^i$. Take the union $\lc{U}$ of all those subsets.  $\lc{U}$  is a  stem in $\lc{C}$. 
$\lc{U}$ is isomorphic to a labeled representative of a finite order, $U$, which is a stem in $C$ and 
has cardinality  $|U|  \le kn $. $U$  is also a certificate of $\Gamma_n$ and since
 $C$ is a minimal certificate, $C = U$ and so $|C|  \le kn $.

If $\Gamma_n$ is a singleton set then its single element is the 
unique minimal certificate of $\Gamma_n$ and $|C| = n$. If $\Gamma_n$ is not a singleton then any minimal certificate must have cardinality greater than $n$. 
\end{proof}

\begin{figure}[htpb]
  \centering
	\includegraphics[width=0.66\textwidth]{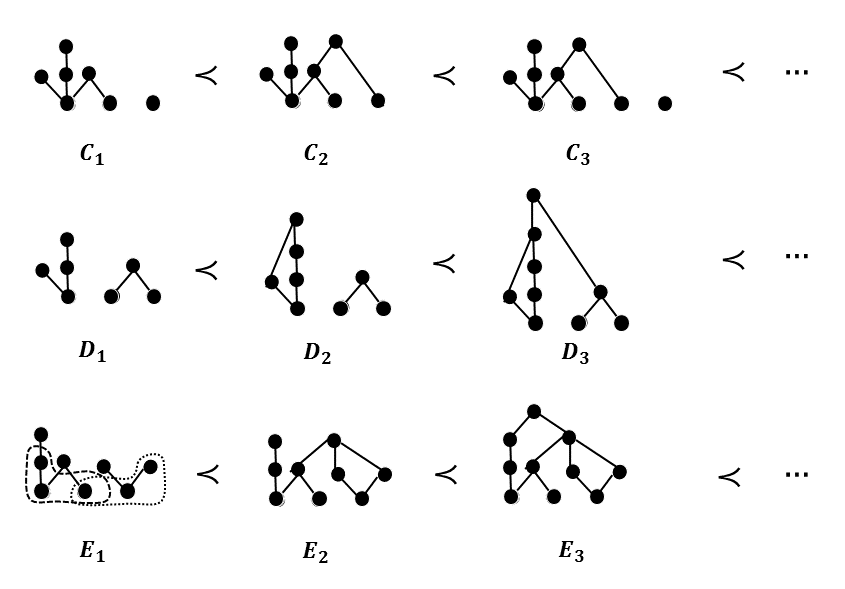}
	\caption{All orders shown in the figure are certificates of $\Omega(3)$. $C_1$, $D_1$ and $E_1$ are minimal certificates of $\Omega(3)$. $C_1$ is a stem in $C_2$, $C_2$ is a stem in $C_3$, and similarly for $D$ and $E$. $|C_1|=|D_1|=7$ and $|E_1|=8$. The dotted outlines on $E_1$ show that the ``L'' order is embedded in $E_1$ in more than one way.}
	\label{min_cert_defns}
\end{figure}

We will also need the concept of a labeled certificate. 
\begin{definition}
A labeled causet $\lc{C}$ is a \textbf{labeled certificate} of $\Gamma_n\in\Lambda$ if $\lc{C}$ is a representative of a certificate of $\Gamma_n$. A labeled causet $\lc{C}$ is a \textbf{labeled minimal certificate} of $\Gamma_n\in\Lambda$ if $\lc{C}$ is a representative of a minimal certificate of $\Gamma_n$.
\end{definition}

An example is shown in figure \ref{labeled_cert_defns}.

\begin{figure}[htpb]
  \centering
	\includegraphics[width=0.44\textwidth]{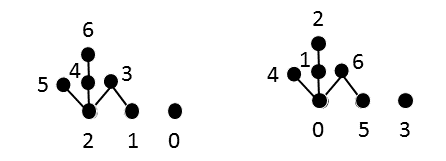}
	\caption{Two labeled minimal certificates of $\Omega(3)$.}
	\label{labeled_cert_defns}
\end{figure}

\subsection{Construction of covtree}\label{subseccovtree}

Given any $\Gamma_n\in\Lambda$, we will be interested in 
 the set of all $k $-stems of elements of $\Gamma_n$ for $k < n$. The following definition will be useful. 
\begin{definition} For any $n$ and any set, $\Gamma_n$, of $n$-orders, the map ${{\mathcal O}}_{-}$  takes $\Gamma_n$ to the set of $(n-1)$-stems of elements of $\Gamma_n$:
\begin{equation} 
{\mathcal O}_{-}(\Gamma_n):=\{B \in \Omega(n-1)\ | \  \exists \ A\in \Gamma_n\ \mathrm{ s.t. }\ B \text{ is a stem in } A \}
\,. \label{o_minus_defn}
\end{equation} 
\end{definition}

One way to think about the operation of ${{\mathcal O}}_{-}$ on $\Gamma_n$  is to take an $n$-order in $\Gamma_n$, choose a maximal element of (a representative of) that $n$-order, and delete that maximal element to form (a representative of) an $(n-1)$-order. The set
${{\mathcal O}}_{-}(\Gamma_n)$ is the set of all $(n-1)$-orders which can be formed in this way. 

\begin{lemma}\label{ominus}
Let $C$ be an $n$-order  and $0<k\le n$. The set of $k$-stems of $C$ is ${{{\mathcal O}}_{-}}^{n-k}(\{\, C \, \})$.

 \end{lemma}
 \begin{proof}
 Consider, $X$, a $k$-stem in $C$. There exist labeled representatives $\lc{X}$ of $X$
and $\lc{C}$ of $C$ such that $\lc{X}$ is a stem in $\lc{C}$.  The ground set of
$\lc{C}$ is the interval $[n-1]$. The $(n-k)$-step process of deleting the elements $n-1$, $n-2$,\dots 
$k$ in turn from $\lc{C}$ results in $\lc{X}$. This shows that $X$  is in 
${{{\mathcal O}}_{-}}^{n-k}(\{\, C \, \})$. And conversely, deleting a maximal element from a representative of $C$  $n-k$ times results in a representative of a $k$-stem of $C$. 
 \end{proof}

\begin{corollary}\label{certdown} Let $\Gamma_{n}$ be a set of $n$-orders. 
 If $C$ is a certificate of $\Gamma_{n}$, then $C$ is also a certificate of ${{\mathcal O}_{-}}^k (\Gamma_{n}) $ for any $k$, $0 \le k < n$. 
 \end{corollary}
 
 The converse is not true: if $C$ is a certificate of ${\mathcal O}_{-}(\Gamma_{n})$, then $C$ may or may not be a certificate of $\Gamma_{n}$. In fact, $\Gamma_{n}$ may have no certificates at all. Examples are shown in figure \ref{o_minus_defn_fig}.

\begin{figure}[htbp]
    \centering
    \begin{subfigure}{12cm}
    \includegraphics[width=.8\textwidth]{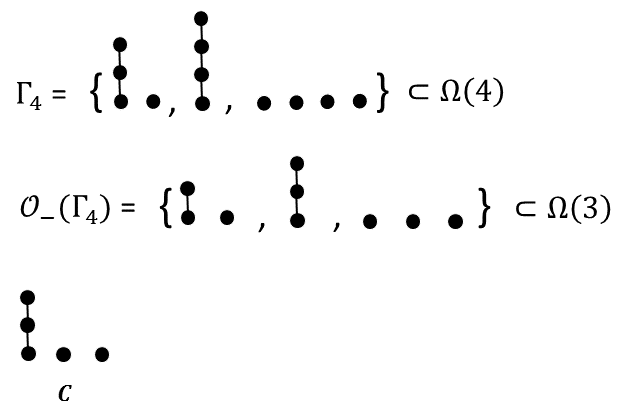}
    \caption{$\Gamma_4$ has no certificates. $C$ is a certificate of ${\mathcal O}_{-}(\Gamma_4)$.}
    \end{subfigure}

    \begin{subfigure}{12cm}
    \includegraphics[width=.8\textwidth]{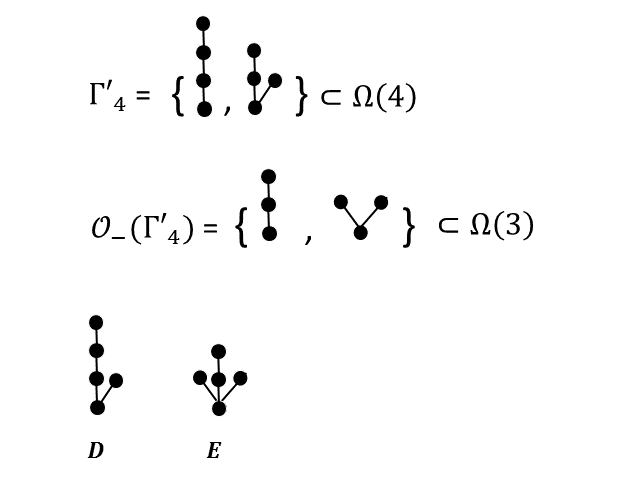}
   \caption{$D$ is a certificate of $\Gamma'_4$, and therefore a certificate of ${\mathcal O}_{-}(\Gamma'_4)$. $E$ is a certificate of ${\mathcal O}_{-}(\Gamma'_4)$ and is not a certificate of $\Gamma'_4$.}
    \end{subfigure}
    \caption{Illustration of the ${{\mathcal O}}_{-}$ operation.}%
    \label{o_minus_defn_fig}%
\end{figure}

We are now ready to define covtree. Recall that $\Lambda$ is the collection of sets of $n$-orders, for all $n$, which have certificates.

\begin{definition} \textbf{Covtree} is the partial order $(\Lambda, \prec)$, where $\Gamma_n\prec\Gamma_{m}$ if and only if $n<m$ and ${\mathcal {O}_{-}}^{m-n}(\Gamma_m)=\Gamma_n$. \end{definition}

Covtree is a tree formed of levels labeled by $\{1,2,3,\dots\}$. The nodes in level $n$ are sets of $n$-orders. A set of $n$-orders, $\Gamma_n$, is a node in level $n$ of covtree if and only if $\Gamma_n$ has a certificate. (This is the motivation for the term certificate: a certificate of $\Gamma_n$ certifies that $\Gamma_n$ is a node in covtree.) 
The partial order on covtree is defined by putting $\Gamma_n$ directly above $O^{-}(\Gamma_n)$, for 
every node $\Gamma_n$, and taking the transitive closure. 

\par  The singleton set containing the one-element order is the root
of covtree.  The first three levels of covtree are shown in figure \ref{covtree}.
There are 22 nodes in level 3 of covtree out of a possible $2^5 - 1 = 31$ non-empty subsets of the set $\Omega(3)$ of $3$-orders. A certificate for each node in level 3 is shown in appendix \ref{certificate_appendix}.  The $9$ ``non-nodes''  in level 3 are given in appendix \ref{certificate_appendix}.

Given a node in covtree, the unique path downwards from it to the root is generated by applying the operator ${\mathcal O}_{-}$ sequentially to the node. 
In particular, every singleton set $\{C\}$ where $C$ is an $n$-order is a node in covtree because $C$ is its certificate and the 
path in covtree down from $\{C\}$ to the root is formed of the nodes ${\mathcal O}_{-}{}^k (\{C\})$, 
$k = 0,1,2,\dots,n-1$. In the upward direction, 
generating the nodes directly above a given $\Gamma_n$ in covtree is a difficult problem.

\begin{figure}[htbp]
    \centering
    \begin{subfigure}{12cm}
    \includegraphics[width=1.2\textwidth]{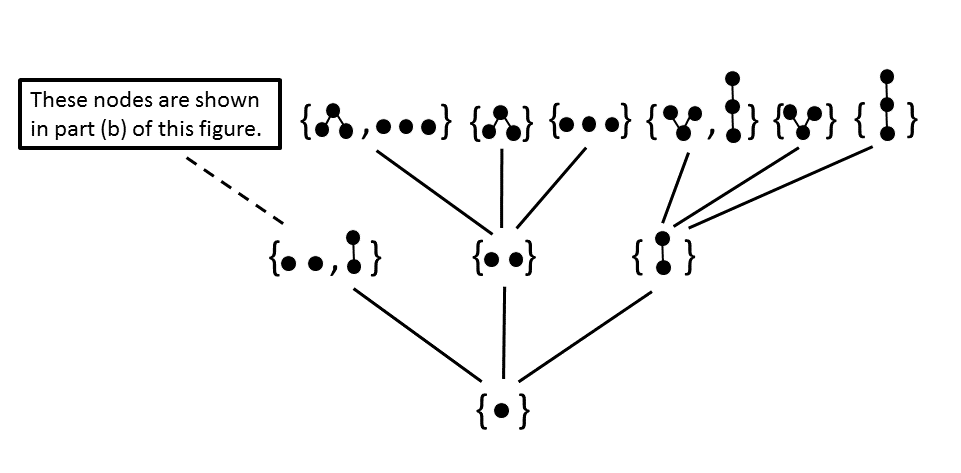}
    \caption{The structure of the first three levels of covtree.}
    \end{subfigure}

    \begin{subfigure}{12cm}
    \includegraphics[width=1.\textwidth]{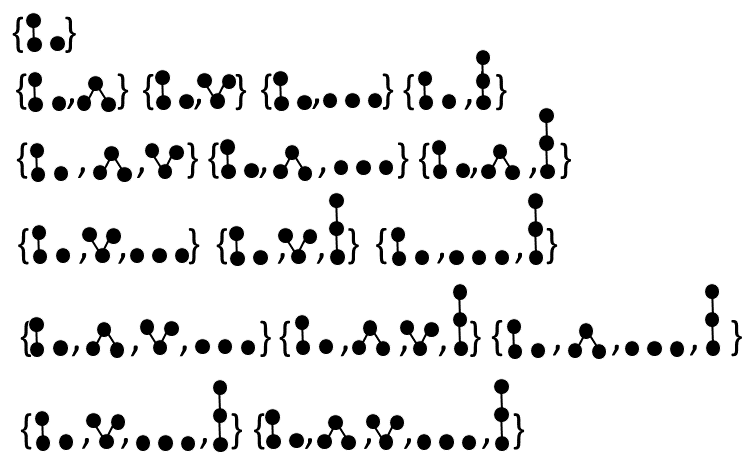}
   \caption{Nodes in level three which are directly above the node $\{$\twoch, \twoach \ \ $\}$.}
    \end{subfigure}
    \caption{The first three levels of covtree.}%
    \label{covtree}%
\end{figure}

\section{Causal set dynamics on covtree}\label{sectheorems}

Covtree allows us to realise the idea described previously of defining a dynamics on a tree in which the nodes in level $n$ are sets of $n$-orders and each node corresponds to the covariant event ``the $n$-stems of the growing order is this set of $n$-orders.''  
Consider a classical dynamical model for a growing causal set as an upward-going random walk 
on covtree, starting at the root.  In preparation for exploring the relationship between paths in covtree  and infinite orders -- the histories in a causal set cosmological model --  we generalise the notion of a certificate of a node to the certificate of a path:
\begin{definition}
An infinite order is a \textbf{certificate of a path} $\mathcal{P}$ in covtree if 
it is a certificate of every node in $\mathcal{P}$.  A \textbf{labeled certificate of a path} $\mathcal{P}$ is a representative of 
a certificate of $\mathcal{P}$.
\end{definition}

One relationship between infinite orders -- elements of ${\Omega}$ -- and paths in covtree  is straightforward to state and understand:
\begin{lemma}\label{claimlevel}
Let  $C$ be an infinite order. The nodes of covtree of which $C$ is a certificate form a path in covtree starting at the root. 
\end{lemma} 
\begin{proof} Let $\Gamma_n$ be the set of $n$-stems of $C$, for each $n>0$. $C$ is a certificate of each $\Gamma_n$.  Each $\Gamma_n$ is a node in covtree and
corollary \ref{certdown} shows that these nodes form a path in covtree down to the root. 
\end{proof}

The map from infinite orders to paths in covtree implied in the lemma above is not one-to-one because a rogue order is not specified by its stems: if $C$ and $C'$ are equivalent rogues then 
they are both certificates of the same path in covtree. This means that our stochastic process cannot, in principle, distinguish between equivalent rogues. 

It is not immediately apparent whether or not \textit{every} infinite path in covtree has an infinite order
as a certificate but in fact it is true and we have:
\begin{theorem}
Let $\mathcal{P}$ be an infinite path in covtree starting at the root. There exists 
an infinite order $C$ which is a certificate of $\mathcal{P}$. 
\label{theorem1}
\end{theorem}

To prove theorem \ref{theorem1} we  will demonstrate  an algorithm to generate a \textit{labeled} certificate of any 
path $\mathcal{P}$. The isomorphism class of this  labeled certificate is then the desired order. We begin with some lemmas. 

\begin{lemma}\label{claim_theorem_1_prep} Let $\mathcal{P}=\{\Gamma_1, \Gamma_2, \Gamma_3\dots\}$ be a path in covtree and let $\Gamma_n\in \mathcal{P}$ not be a singleton. Then there exists a node in $\mathcal{P}$ above $\Gamma_n$ that contains a certificate of $\Gamma_n$ as an element.  
In other words, there exists an $m> n$ and 
an $m$-order $C$ such that  $C$ is a certificate of $\Gamma_n$ and $C\in \Gamma_m\in\mathcal{P}$. 
\end{lemma}

\begin{proof} By lemma {\ref{cardmincert}},  the cardinality of any minimal certificate, $C$, of $\Gamma_n$ satisfies $n< \vert C\vert \le N$ where $N := n |\Gamma_n|$. Consider $\Gamma_N\in \mathcal{P}$ and let $D$ be a finite certificate of $\Gamma_N$.  By corollary \ref{certdown},  $D$ is a certificate of every node below $\Gamma_N$ so 
$D$ is a certificate of $\Gamma_n$.  
Now, at  least one minimal certificate of $\Gamma_n$ occurs as a stem in $D$. Choose one, call it $C$, let $m := \vert C \vert$ and consider 
$\Gamma_m\in \mathcal{P}$.  $C$ is an $m$-stem in $D$.  $\Gamma_m$
is the set of all $m$-stems of $D$ and so $C$ is an element of $\Gamma_m$.  
\end{proof}

Note the choices made in the proof above: a choice of a particular certificate of $\Gamma_N$ and a choice of a stem in it which is a minimal certificate of $\Gamma_n$.

\begin{lemma}\label{claim4_theorem_1_prep} Let $\mathcal{P}=\{\Gamma_1, \Gamma_2, ...\}$ be a path in covtree and let $\Gamma_n\in \mathcal{P}$. There is a node in $\mathcal{P}$ above $\Gamma_n$
 which has a certificate of $\Gamma_n$ as an element. 
\end{lemma}
\begin{proof}
 In the case that $\Gamma_n$ is not a singleton,  a node with the required property 
is $\Gamma_m$ as defined in the proof of lemma \ref{claim_theorem_1_prep}. 
In the case that $\Gamma_n$ is a singleton, then 
a node with the required property is $\Gamma_{n+1}$ because every element of $\Gamma_{n+1}$ is a certificate of $\Gamma_{n}$. 
\end{proof}

\begin{lemma}\label{claim3_theorem_1_prep}
Let  $\mathcal{P}=\{\Gamma_1, \Gamma_2, ...\}$ be an infinite path in covtree. There exists an infinite subsequence 
$m_1< m_2< m_3<\dots$ of the natural numbers, and a set of labeled causal sets $\{\lc{C}_{m_1}, \lc{C}_{m_2}, \lc{C}_{m_3} \dots\}$ 
such that, for all $k$, 
\begin{itemize}
\item[(i)] $ \vert \lc{C}_{m_k} \vert = m_{k} $; 
\item[(ii)] $ \lc{C}_{m_k} $ is a subcauset, a stem, in $ \lc{C}_{m_{k+1}}$; 
\item[(iii)] $\lc{C}_{m_{k+1}}$ is a labeled certificate of $\Gamma_{m_k}$ and also therefore a labeled certificate
of all nodes below 
$\Gamma_{m_k}$;
\item[(iv)] ${C}_{m_{k+1}}$, the isomorphism class of  $\lc{C}_{m_{k+1}}$, 
 is an element of $\Gamma_{m_{k+1}}$.
\end{itemize}
\end{lemma}

\begin{proof} \label{proof_theorem_1}
Consider an infinite path $\mathcal{P}=\{\Gamma_1, \Gamma_2, ...\}$ in covtree. The required sequence of causal sets $\{\lc{C}_{m_1}, \lc{C}_{m_2}, \lc{C}_{m_3} \dots\}$ is constructed by the following inductive algorithm.
\newline \textit{ Step 1:}
\newline 1.0) Pick some nonzero natural number $m_0$ to start and 
consider $\Gamma_{m_0}\in \mathcal{P}$.
\newline 1.1) By  lemma \ref{claim4_theorem_1_prep}  there exists an $m_1$ such that $m_1 > m_0$ 
and such that $\Gamma_{m_1} $
contains a certificate of $\Gamma_{m_0}$ as an element. Call that certificate $C_{m_1}$.
Its cardinality is $\vert C_{m_1} \vert = m_1$. 
\newline 1.2) Pick $\lc{C}_{m_1}$,  a labeled causet which is a representative of $C_{m_1}$.
\newline 1.3) Go to step 2.
\newline \textit{ Step} $k>1$:
\newline k.1)   By  lemma \ref{claim4_theorem_1_prep}  there exists an $m_{k}$ such that $m_{k} > m_{k-1}$ 
and such that $\Gamma_{m_{k}} $
contains a certificate of $\Gamma_{m_{k-1}}$ as an element. Call that certificate $C_{m_k}$.
Its cardinality is $\vert C_{m_k} \vert = m_{k}$.  
\newline k.2)  Consider $C_{m_{k-1}}$ and its labeled representative $\lc{C}_{m_{k-1}}$ from the previous step. 
$C_{m_{k-1}}$  is an element of $\Gamma_{m_{k-1}}$.
Because $C_{m_k}$ is a certificate of 
$\Gamma_{m_{k-1}}$,  $C_{m_{k-1}}$ is an $m_{k-1}$-stem of $C_{m_{k}}$. 
Pick a representative  $\lc{C}_{m_k}$ of  ${C}_{m_k}$ such that 
$\lc{C}_{m_{k-1}}$ from the previous step is a sub-causet of $\lc{C}_{m_{k}}$.
\newline k.3) Go to step $k+1$.
\newline The subsequence $m_1< m_2< m_3<\dots$ of the natural numbers, and the set of labeled causal sets $\{\lc{C}_{m_1}, \lc{C}_{m_2}, \lc{C}_{m_3} \dots\}$ have the required properties by construction. 
 \end{proof}
 
 \begin{lemma} An infinite path in covtree has a labeled certificate. 
\end{lemma}
\begin{proof} The union of the nested sequence of labeled causets $\{\lc{C}_{m_1}, \lc{C}_{m_2}, \lc{C}_{m_3} \dots\}$ of the previous lemma is a labeled certificate of the path. 
\end{proof}
\begin{corollary}\label{corollary2}
An infinite path in covtree has a certificate. 
\end{corollary}
This corollary is theorem \ref{theorem1}.

\subsection{Measures on $\mathcal{R}(\mathcal{S})$}\label{subsec_measures_on_R(S)}

We propose random walks upwards on covtree as dynamical models in which an order grows and in which arriving at a node $\Gamma_n$ corresponds to the occurrence of the event ``the set of $n$-stems of the order is $\Gamma_n$.''
A question that arises is: what is the relationship between dynamical models on covtree and dynamical models on labeled poscau? Do the kinematical structures of covtree and labeled poscau give rise to different classes of causet growth models?
We will show that the set of measures induced on $\mathcal{R}(\mathcal{S})$ by walks on labeled poscau and the set of measures induced on $\mathcal{R}(\mathcal{S})$ by walks on covtree are equal, and equal to the set of all measures on $\mathcal{R}(\mathcal{S})$.

\par First we introduce the covtree measure space. The certificate set, $cert(\Gamma_n)$, of a node, $\Gamma_n$, in covtree is the subset \begin{equation} cert(\Gamma_n):=\{ C\in\Omega \mid C \text{ is a certificate of } \Gamma_n\}. \end{equation}

The node certificate sets are the covtree ``cylinder sets''. Let $\Sigma$ denote the set of node certificate sets $cert(\Gamma_n)$ for all nodes in covtree, together with the empty set. A random walk on covtree, defined by the transition probabilities for each link in covtree satisfying the Markov sum rule, gives a measure $\mu$ on $\Sigma$, where $\mu(cert(\Gamma_n))$ is the product of the transition probabilities on the links of the path from the root to $\Gamma_n$. The tree structure of covtree means that $\Sigma$ is a semi-ring and that a measure $\mu$ on $\Sigma$ generated by a set of Markovian transition probabilities on covtree is countably-additive \footnote{This is standard measure theory for stochastic processes. For completeness, we present proofs in appendix \ref{appendix_subsec_R(S)}.}.
Hence we can apply the Fundamental Theorem of Measure Theory \cite{Kolmogorov:1975} which says that the measure $\mu$ extends to $\mathcal{R}(\Sigma)$, the $\sigma$-algebra generated by $\Sigma$.

We are now in a position to prove that:

\begin{lemma}\label{lemma_100901} $\mathcal{R}(\mathcal{S})=\mathcal{R}(\Sigma)$. \end{lemma}
\begin{proof}
\par First we note that as defined, these two $\sigma$-algebras are defined over different 
sample spaces: an element of $\Sigma$ is a set of infinite orders and an element of $\mathcal{S}$ is a set of infinite labeled causets. However, both the covariant algebra $\mathcal{R}$ and the stem algebra $\mathcal{R}(\mathcal{S})$ can be thought of, in an obvious way, as $\sigma$-algebras on the sample space ${\Omega}$ of infinite orders, since their elements are covariant. This is the sense in which the claim is to be interpreted.
\par We will show that any stem set -- thought of as a set of infinite orders -- can be constructed by finite set operations on the certificate sets and vice versa, and the result follows.
\par Consider an $n$-order  $B$. Let $\Gamma_n^i$ be the nodes in covtree such that $B\in \Gamma_n^i$, where $i$ labels the nodes. Suppose $C\in cert(\Gamma_n^i)$ for some $i$. Then $B$ is a stem in $C$ and hence $C\in stem(B)$. Suppose $C\notin cert(\Gamma_n^i)$ for all $i$. Then $B$ is not a stem in $C$ and hence $C\notin stem(B)$. It follows that $stem(B)=\bigcup_i cert(\Gamma_n^i)$.
\par Consider some node $\Gamma_n=\{A^1, ..., A^k\}$ in covtree. Let $\Omega(n)\setminus \Gamma_n =\{B^1, ..., B^l\}$. Suppose $C\in cert(\Gamma_n)$. Then $A^1,...,A^k$ are stems in $C$, and $B^1,...,B^l$ are not stems in $C$. Hence $C\in \bigcap \limits_{i=1}^k stem(A^i)\setminus \bigcup \limits_{j=1}^l stem(B^j)$. Suppose $C\notin cert(\Gamma_n)$. Then either $(i)$ there exists some $A^i\in \Gamma_n$ which is not a stem in $C \implies C\notin \bigcap \limits_{i=1}^k stem(A^i)$, or $(ii)$ there exists some $B^j\in \Omega(n)\setminus \Gamma_n$ which is a stem in $C \implies C\in \bigcup \limits_{j=1}^l stem(B^j)$. It follows that, $cert(\Gamma_n)=\bigcap \limits_{i=1}^k stem(A^i)\setminus \bigcup \limits_{j=1}^l stem(B^j)$.
\end{proof}

Hence every walk on covtree induces a unique measure on $\mathcal{R}(\mathcal{S})$, and every measure on $\mathcal{R}(\mathcal{S})$ induces a unique walk on covtree: the transition probability in the covtree walk from node $\Gamma_n$ to node $\Gamma_{n+1}$ directly above it is the measure of $cert(\Gamma_{n+1})$ divided by the measure of $cert(\Gamma_{n})$. Therefore, let us call a measure on $\mathcal{R}(\mathcal{S})$ a \textbf{covtree measure}.

By a similar argument to the above, there is  is a 1-1 correspondence between walks on labeled poscau and measures on $\tilde{\mathcal{R}}$ so we will call a measure on $\mathcal{R}(\mathcal{S})$ a \textbf{poscau measure} if it is a restriction to $\mathcal{R}(\mathcal{S})$ of some measure $\tilde{\mu}$ on $\tilde{\mathcal{R}}$. A \textbf{CSG measure} on $\mathcal{R}(\mathcal{S})$ is a poscau measure such that $\tilde{\mu}$ is induced by a CSG walk.

It follows from lemma \ref{lemma_100901} that every poscau measure on $\mathcal{R}(\mathcal{S})$ is a covtree measure on $\mathcal{R}(\mathcal{S})$. In fact, it is also true that every covtree measure is a poscau measure:

\begin{lemma} \label{lemma040901} For every measure $\mu$ on $\mathcal{R}(\mathcal{S})$ there exists an extension $\tilde{\mu}$ to $\tilde{\mathcal{R}}$.\end{lemma}

\begin{proof}
First note that there is a metric on $\tilde{\Omega}$ with respect to which $(\tilde{\Omega}, \tilde{\mathcal{R}})$ is a Polish space \cite{Brightwell:2002vw}. Since every Polish space is a Lusin space \cite{Schwartz:1973}, $(\tilde{\Omega}, \tilde{\mathcal{R}})$ is a Lusin space. 
Note also that $\mathcal{R}(\mathcal{S})$ is a separable sub-$\sigma$-algebra of $\tilde{\mathcal{R}}$ since there exists a countable collection of subsets of $\tilde{\Omega}$ which generates $\mathcal{R}(\mathcal{S})$, namely $\mathcal{S}$ (or $\Sigma$).
The result follows from the theorem that if $(Y, \mathcal{B})$ is a Lusin space, then every measure defined on a separable
sub-$\sigma$-algebra of $\mathcal{B}$ can be extended to $\mathcal{B}$ \cite{Landers:1974}.
\end{proof}

\section{Discussion} \label{secdiscussion}

At the beginning of section \ref{seccovtree} we posed the question: ``Is it possible to construct a physically well-motivated measure on the stem algebra $\mathcal{R}(\mathcal{S})$ \textit{directly}, in a manifestly label-independent way that does not rely on any gauge dependent notion?" 
The key phrase here is \textit{physically well-motivated}. We have shown that we can generate a mathematically well-defined measure on the stem events $\mathcal{R}(\mathcal{S})$ via a growth process conceived as a random walk up covtree. There is no reason to expect, however, that a generic such walk will be physically interesting: the class of walks is too vast to be interesting. We need physically motivated conditions to restrict the models to a sub-class worth studying. This is what was done by Rideout and Sorkin in the context of walks up labeled poscau by 
imposing the conditions of discrete general covariance (DGC) and Bell causality (BC) \cite{Rideout:1999ub}. These conditions restrict the class of walks on labeled poscau to the CSG models.

The relationship between the ``labeled'' conditions of DGC and BC and any conditions on covtree walks is not understood. Note that lemma \ref{lemma040901} means that although every covtree walk is apparently completely covariant in its setup, for every walk on labeled poscau -- whether it satisfies Discrete General Covariance or not -- there exists a covtree walk that produces the same measure on $\mathcal{R}(\mathcal{S})$. So, 
there is no easy relationship between the DGC condition on a labeled poscau walk and the manifest ``covariance'' of a covtree walk. 
We can frame the sort of progress we'd like to make from here as a set of interrelated questions.\begin{itemize}
\item[ (i) ] Is there a condition on the transition amplitudes of a walk up covtree such that the covtree measure is a poscau measure from a walk on labeled poscau that satisfies DGC only, measures which Brightwell and Luczak call ``order-invariant''? \cite{Brightwell:2011,Brightwell:2012,Brightwell:2016}. 
\item[(ii)]  Is there a condition on the transition amplitudes of a walk up covtree such that the covtree measure equals a CSG measure? 
\item[(iii)] Is there a condition on a random walk up covtree which expresses the physical condition of relativistic causality? How is this related to the condition of Bell Causality satisfied by CSG models as walks on labeled poscau? Is this new condition enough to reduce the class to a physically interesting one or are other conditions needed and what are they?
\item[(iv)] What is the role of the rogues, if any, in understanding the physics of covtree walks? Could the condition that the set of rogues has measure zero -- as it does for any CSG model -- be considered as a physical condition in itself and what conditions on the transition amplitudes for the walk would imply this condition? 
\item[(v)] What form might a quantum random walk on covtree take and might it be possible to formulate a quantum relativistic causality condition for it, even while the labeled BC condition has thus far resisted a quantal generalisation?
\end{itemize} 

Here we start to grapple with the kinds of knotty questions that crop up when considering what a condition of relativistic causality might look like in a theory in which the spacetime causal order itself is dynamical and stochastic/quantal and in which labels/coordinates are banned, even as a prop to kick away at the end. Here, in covtree, at least we now have a concrete arena in which to investigate these questions. 

\par \textbf{Acknowledgments:} We thank Jeremy Butterfield for useful discussions. This research was supported in part by Perimeter Institute for Theoretical Physics. Research at Perimeter Institute is supported by the Government of Canada through Industry Canada and by the Province of Ontario through the Ministry of Economic Development and Innovation. FD is supported in part by STFC grant ST/P000762/1 and APEX grant APX/R1/180098. SZ thanks the Perimeter Institute and the Raman Research Institute for hospitality while this work was being completed. SZ is partially supported by the Kenneth Lindsay Scholarship Trust. 

\appendix
\section{Table of sets defined in the text}

\begin{tabular}{| l | l |}
\hline
  $\tilde{\Omega}(n)$ & The set of labeled causets of cardinality $n$ \\ \hline
  $\tilde{\Omega}(\mathbb{N})$ & The set of finite labeled causets \\ \hline
  $\tilde{\Omega}$ & The set of infinite labeled causets \\ \hline
  $\Omega(n)$ & The set of $n$-orders  \\ \hline
  $\Omega(\mathbb{N})$ & The set of finite orders \\ \hline
  $\Omega$ & The set of infinite orders \\ \hline
  $cyl(\lc{A})$ &  $cyl(\lc{A})=\{ \lc{C}\in\tilde{\Omega} \mid \lc{A} \text{ is a stem in } \tilde{C}\}$\\ \hline
 	$\tilde{\mathcal{R}}$ & The $\sigma$-algebra generated by the cylinder sets \\ \hline
 	$\mathcal{R}$ & The covariant sub-$\sigma$-algebra of $\tilde{\mathcal{R}}$ \\ \hline
	$stem(A)$ & $stem(A) = \{  \lc{C} \in \tilde{\Omega} \ |\  A \ \textrm{is a stem in} \ \lc{C} \}$\\ \hline
	$\mathcal{S}$ & The set of stem sets \\ \hline
	$\mathcal{R}(\mathcal{S})$ & The sub-$\sigma$-algebra of $\tilde{\mathcal{R}}$ generated by $\mathcal{S}$ \\ \hline
	$\Gamma_n$ &  A subset of $\Omega(n)$\\ \hline
	$cert(\Gamma_n)$ & $cert(\Gamma_n)=\{ C\in\Omega \mid C \text{ is a certificate of } \Gamma_n\}$ \\ \hline
  $\Lambda$ &  $\Lambda = \bigcup\limits_{n\in\mathbb{N}^+}\{\Gamma_n \subseteq \Omega(n)| \exists \text{ a certificate for } \Gamma_n\}$\\ \hline
	$\mathcal{P}$ & An infinite path from the root in covtree \\ \hline
	$\Sigma$ & The set of certificate sets together with the empty set \\ \hline
	$\mathcal{R}(\Sigma)$ & The sub-$\sigma$-algebra of $\tilde{\mathcal{R}}$ generated by $\Sigma$ \\ \hline
\end{tabular}
\pagebreak

\section{Proofs}\label{appendix_proofs}
\subsection{Appendix to section \ref{subseccovtree}}\label{certificate_appendix}

\begin{figure}[htpb]
  \centering
	\includegraphics[width=0.66\textwidth]{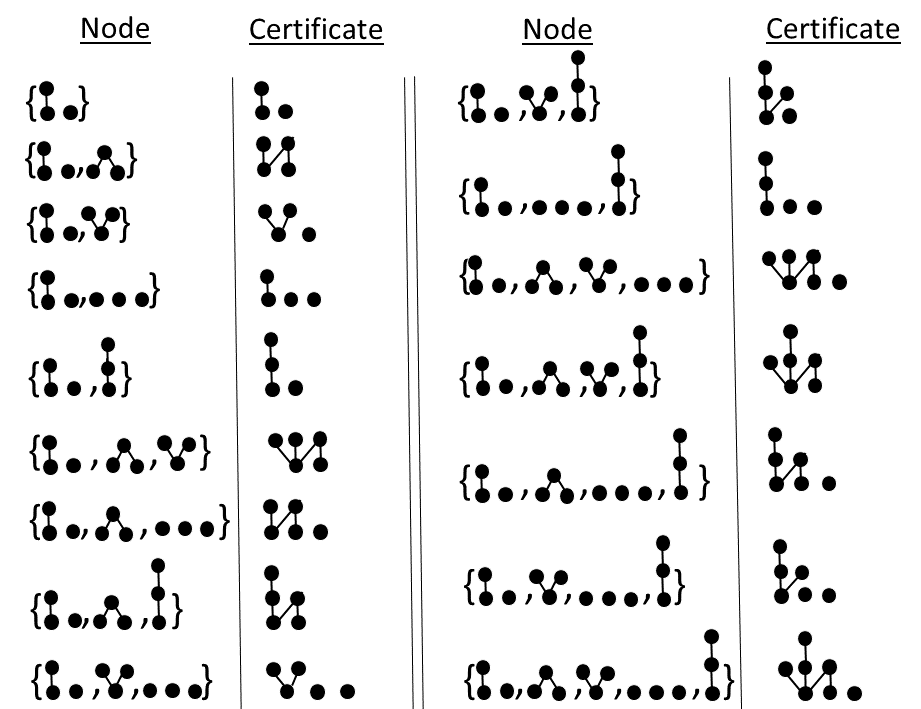}
	\caption{The level 3 nodes which are directly above the level 2 doublet are shown together with their respective certificates.}
	\label{appendix_certificates_level_3}
\end{figure}

\begin{figure}[htpb]
  \centering
	\includegraphics[width=0.6\textwidth]{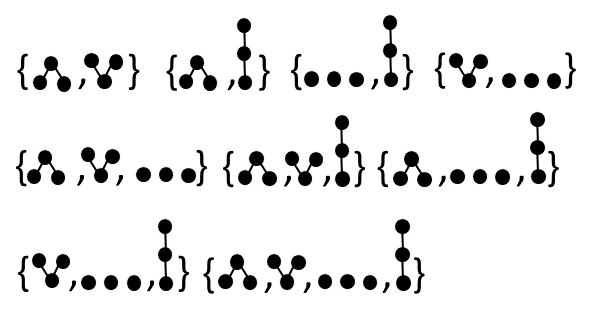}
	\caption{The sets shown in the figure have no certificates and therefore are not nodes. For every set shown, if an order contains all the elements of that set as stems then it also contains the L as a stem.}
	\label{nonnodes}
\end{figure}

\newpage

\subsection{Appendix to section \ref{sectheorems}}

\begin{proof}[Alternative proof to theorem \ref{theorem1}]

Recall that $\sim_R$ denotes the rogue equivalence relation, and let $p:\Omega\rightarrow\Omega/\sim_R$ be the associated canonical quotient map. Let $[A]_R$ denote an element of $\Omega/\sim_R$, where $A\in\Omega$ is a representative of $[A]_R$.

Define the following metric on $\Omega/\sim_R$: $d([A]_R,[B]_R)=\frac{1}{2^n}$, where $n$ is the highest integer such that the set of $n$-stems of $A$ is the set of $n$-stems of $B$. One can show that $(\Omega/\sim_R,d)$ is a complete metric space.

Let $[cert(\Gamma_n)]_R\subset \Omega/\sim_R$ denote the image of the certificate set $cert(\Gamma_n)\subset\Omega$ under the quotient map $p$. Then one can show that $[cert(\Gamma_n)]_R$ is both open\footnote{$[cert(\Gamma_n)]_R$ are exactly the open balls under the metric topology.} and closed in $(\Omega/\sim_R,d)$.

The diameter of $[cert(\Gamma_n)]_R$, $d([cert(\Gamma_n)]_R)$, is defined to be the maximum distance between any two elements in $[cert(\Gamma_n)]_R$ and is equal to $\frac{1}{2^{n}}$. Hence $d([cert(\Gamma_n)]_R)\rightarrow 0$ as $n\rightarrow\infty$.

Given a path in covtree, $\mathcal{P}=\{\Gamma_1, \Gamma_2, ...\}$, then the following is a nested sequence: $[cert(\Gamma_1)]_R\supset [cert(\Gamma_2)]_R \supset\dots$.

Now, Cantor's Lemma states that a metric space $(X, d)$ is complete if and only if, for every nested sequence $\{F_n\}_{n\geq 1}$ of nonempty closed subsets of $X$, that is, (a) $F_1\supseteq 	 F_2\supseteq  \dots$ and (b) $d(F_n)\rightarrow 0$ as $n\rightarrow\infty$, the intersection $\bigcap_{n=1}^{\infty} F_n$ contains one and only one point \cite{Shirali:2006}. 
\end{proof}

\subsection{Appendix to section \ref{subsec_measures_on_R(S)}}\label{appendix_subsec_R(S)}

Recall that $\Sigma$ is the collection of certificate sets with the empty set.
\begin{lemma} \label{semiring_proof}
$\Sigma$ is a semi-ring.
\end{lemma}

\begin{proof}
\par A family $\mathcal{F}$ of subsets of a set $\mathcal{M}$ is a semi-ring if $(i)$ $\emptyset \in  \mathcal{F}$, $(ii)$ $A\cap B\in \mathcal{F}$ for all $A, B \in \mathcal{F}$, and $(iii)$ for every pair of sets $A, B \in \mathcal{F}$ with $A\subset B$, the set $B\setminus A$ is the union of finitely many disjoint sets in $\mathcal{F}$ \cite{Kolmogorov:1975}.

\par Let $\Gamma_m$ and $\Gamma_n$ be nodes in covtree, $m>n$. Suppose $\Gamma_n\prec \Gamma_m$. Then $cert(\Gamma_m)\subset cert(\Gamma_n) \implies cert(\Gamma_m)\cap cert(\Gamma_n)=cert(\Gamma_m)\in \Sigma$. Suppose $\Gamma_n\not\prec \Gamma_m$. Then $cert(\Gamma_m)\cap cert(\Gamma_n)=\emptyset \in \Sigma$.
\par Let $cert(\Gamma_m)\subset cert(\Gamma_n)$. Then $cert(\Gamma_n)\setminus cert(\Gamma_m)$ is the set of all certificates of $\Gamma_n$ which are not certificates of $\Gamma_m$. Let $\Gamma_m^i$, 
$i = 1,2,\dots k$,  denote the nodes in level $m$ such that $\Gamma_m^i\succ \Gamma_n$ and $\Gamma_m^i \not= \Gamma_m$. Then $cert(\Gamma_n)\setminus cert(\Gamma_m)=\bigsqcup_i cert(\Gamma_m^i)$.
\end{proof}

Recall that a random walk on covtree, defined by the transition probabilities for each link in covtree satisfying the Markov sum rule, gives a measure $\mu$ on $\Sigma$: $\mu(cert(\Gamma_n)) = $ the product of the transition probabilities on the links of the 
path from the root to $\Gamma_n$ (and $\mu(\emptyset)=0)$.
\begin{lemma} \label{claim6.2}
The measure $\mu$ on $\Sigma$ is  countably-additive.
\end{lemma}
\begin{proof}
We defined $\mu:\Sigma \rightarrow [0,1]$ by $\mu(cert(\Gamma_m))=\mathbb{P}(\Gamma_m)$ where $\mathbb{P}(\Gamma_m)$ is the probability of a random walk to pass through $\Gamma_m$. Also $\mu(\emptyset)=0$. 

\par Suppose $cert(\Gamma_n)=\bigsqcup_{i=1}^k cert(\Gamma^i_{n_i})$, where $i$ labels the individual nodes. Then $\mu(\bigsqcup_{i=1}^k cert(\Gamma^i_{n_i}))=\mu(cert(\Gamma_n))=\mathbb{P}(\Gamma_n)=\sum_{i}\mathbb{P}(\Gamma^i_{n_i})$. Hence, finite additivity is satisfied.

\par Next we will show that countable additivity of $\mu$ is trivially satisfied as no certificate set is a countable disjoint union of certificate sets\footnote{This is a special case of the countable additivity property of cylinder sets associated with a Markov process on a directed finite-valency tree. Another example in the context of the measure space of coin-tosses is given in \cite{Lindstrom:2017}.}. Consider some $\Gamma_m$ in covtree and suppose for contradiction that $cert(\Gamma_m)=\bigsqcup_{i\in\mathbb{N}}cert(\Gamma^i_{n_i})$. Consider the following suborder in covtree, $\{\Gamma_n \in \Lambda | \Gamma_n\succeq \Gamma_m $ and $ \Gamma_n\not\succ \ \Gamma_{n_i}^i \forall \ i\}$, and let $T_m$ be the transitive reduction of it.
\par We note that $(i)$ $T_m$ is infinite, $(ii)$ every node in $T_m$ has finite valency, and $(iii)$ $T_m$ is a connected tree.
\par Then by K\"{o}nig's lemma, $T_m$ contains an infinite upward-going path starting at $\Gamma_m$ \cite{Wilson:1985}. It follows that there is an infinite path $\mathcal{P}$ in covtree such that $\Gamma_m\in \mathcal{P}$ and $\Gamma_{n_i}^i\notin \mathcal{P}$ for all $i\in\mathbb{N}$. Therefore there exists a certificate $C$ of $\mathcal{P}$ and hence of $\Gamma_m$ such that $C\notin cert(\Gamma_{n_i}^i)$ for all $i\in\mathbb{N}$, which is a contradiction.
\end{proof}

\bibliography{covtree}{}
\bibliographystyle{unsrt}

\end{document}